\documentclass[fleqn]{elsart5p}

\usepackage{amscd, amsmath, amssymb, epic, eepic,  delarray,stmaryrd}
\usepackage{lscape}
\usepackage{pstricks,pst-node}

\usepackage[ps,curve,2cell]{xypic}

\renewcommand{\marginpar}[1]{}

\newcommand{\myarrow}[3]{\xymatrix@1{#1 \ar[r]^{#2} & #3}}

\usepackage[ps,curve,2cell]{xypic}

\renewcommand{\marginpar}[1]{}

\newtheorem{definition}{Definition}
\newtheorem{theorem}{Theorem}

\newtheorem{proposition}{Proposition}
\newtheorem{corollary}{Corollary}
\newtheorem{exemple}{Example}
\newtheorem{Proof}{\it Proof.}

\newenvironment{proof}{\begin{Proof}\rm } {\hfill $\Box$
\end{Proof}}

\newcommand{\Jamal}[1]{\textcolor{black}{#1}}

\begin{document}
\begin{frontmatter}



\title{A finite basis theorem for the description logic ${\cal ALC}$}

\author[France1]{Marc Aiguier\corauthref{cor}},
\corauth[cor]{Corresponding author.}
\ead{marc.aiguier@centralesupelec.fr}
\author[France2]{Jamal Atif},
\ead{jamal.atif@dauphine.fr}
\author[France3]{Isabelle Bloch},
\ead{isabelle.bloch@telecom-paristech.fr}
\author[France1]{C\'eline Hudelot},
\ead{celine.hudelot@centralesupelec.fr}

\address[France1]{MICS, Centrale Supelec, Universit\'e Paris-Saclay, France}
\address[France2]{PSL, Universit\'e Paris-Dauphine, LAMSADE, UMR 7243, France}
\address[France3]{LTCI, CNRS, T\'el\'ecom ParisTech, Universit\'e Paris-Saclay, Paris, France}

\begin{abstract}



The main result of this paper is to prove the existence of a finite basis in the description logic ${\cal ALC}$.  We show that  the set of General Concept Inclusions (GCIs) holding in a finite model has always a finite basis, i.e. these GCIs can be derived from finitely many of the GCIs. This result extends a previous result from Baader and Distel, which showed the existence of a finite basis for GCIs holding in a finite model but for the 
inexpressive description logics ${\cal EL}$ and ${\cal EL}_{gfp}$. We also provide an algorithm for computing this finite basis, and prove its correctness. As a byproduct,  we extend our finite basis theorem to any finitely generated complete covariety (i.e. any class of models closed under morphism domain, coproduct and quotient, and generated from a finite set of finite models). 
\end{abstract}

\end{frontmatter}



\section{Introduction}

Description logics (DLs)~\cite{handbook03} are a family of logic-based knowledge representation formalisms that originate from early knowledge representation systems such as frame-based systems~\cite{Min81} and semantics networks~\cite{Sow91}.  Briefly, theories in DL, so-called knowledge bases, are sets of general concept inclusion axioms (GCIs) of the form $C \sqsubseteq D$ where $C$ and $D$ are concepts, \Jamal{i.e.} expressions freely generated from a set of basic concept names and both operators in  $\{\sqcap,\sqcup,\_^c\}$ and quantifiers in $\{\forall r, \exists r\}$ where $r$ is a binary relation name.
 The DL thus defined is often called ${\cal ALC}$. Both extensions and restrictions of ${\cal ALC}$ have been proposed. Among its restrictions, we have the DL ${\cal EL}$ and its extension ${\cal EL}_{gfp}$ to cyclic concept definitions interpreted with greatest fixpoint semantics. ${\cal EL}$ and ${\cal EL}_{gfp}$ restrict the syntax to the operator $\sqcap$ and the quantifier $\exists r$. Although quite inexpressive, the DLs ${\cal EL}$ and ${\cal EL}_{gfp}$ have good features to allow for efficient reasoning procedures~\cite{BBL05,BLS06}.  Baader and Distel have then shown for ${\cal EL}$ and its extension ${\cal EL}_{gfp}$ that the set of GCIs holding in a finite model always has a finite basis~\cite{Dis11,BD08,BD09}, i.e. a finite subset of GCIs from which 
all the others can be derived. They obtained this result by using methods from formal concept analysis~\cite{GW97}. In this paper, we propose to extend this result for the DL ${\cal ALC}$. For this, we learn from Birkhoff's result established in 50's~\cite{BS81} which shows that there is a finite basis for any finite model in universal algebra. We will further give a simple condition on finite models to effectively build such a finite basis.  Thus, we answer an open problem in~\cite{BD08} but which has not received a positive answer yet to our knowledge. \\ We also propose to extend this first result to any finitely generated complete covarieties of models, i.e. any class of models closed under morphism domain, coproducts and quotients\footnote{This notion of complete covariety is lower than \Jamal{the one} defined in~\cite{GS01}. In~\cite{GS01}, complete covariety is any class of models closed under subcoalgebras, coproduct, quotient and total bisimulations. We will show in Section~\ref{morphism and bisimulation} that the kind of morphisms we consider are functional bisimulations.}, and finitely generated from a finite set of finite models. We will first show results of GCI preservation for complete covarieties\footnote{This is not so surprizing because of the relationship between $\mathcal{ALC}$ and the modal logic~\cite{BHS08}.}. Then, we will show that any complete covariety generated from a finite set of finite models has a finite basis.

\medskip
The paper is organized as follows. In Section~\ref{description logic}, we recall basic definitions and notations about the DL ${\cal ALC}$. We present both descriptive and fixed point semantics, the latter being used to give a meaning to cyclic concept descriptions. The presentation of semantics for cyclic concept descriptions slightly differs from the one traditionally found in DL papers~\cite{Dis11}, which iterates on interpretations to find the expected fixed points. Here, we will iterate on the set of individuals by applying a method that can be compared to the one used to interpret formulas in fixpoint logics~\cite{AN01,Koz83}, as already observed in~\cite{Schild94}.  
In Section~\ref{morphism and bisimulation}, we establish some links between model morphisms and bisimulations, and show how GCIs are preserved under particular morphisms. The aim of this section is to provide basic constructions that are useful for establishing the fundamental result of this paper, and its extension to complete covarieties. In Section~\ref{A finite base theorem} we state and prove the existence of a finite basis theorem for the description logic $\mathcal{ALC}$. This  result is similar to Birkhoff's theorem in universal algebra that proves that, for every finite model $\mathcal{I}$, there is a finite basis. An algorithm for computing such a finite basis is provided, and its theoretical guarantees are discussed, in Section~\ref{algorithm}. In Section~\ref{basic constructions} we extend our finite basis theorem to complete covarieties. To this aim, we give a result of characterization of classes of models similar to Rutten's covariety theorem for coalgebras (cf. Theorem 15.3 in~\cite{RuttUni}) that states that any covariety is determined by a subcoalgebra of the final one.


\section{The DL ${\cal ALC}$}
\label{description logic}

\subsection{Syntax}

Concept descriptions are built from a set $N_C$ of concept names and a set $N_R$ of role names which form the signature $\Sigma = (N_C,N_R)$. 

\begin{definition}[Concept descriptions]
\label{concept expressions}
Let $\Sigma = (N_C,N_R)$ be a signature. The set of {\bf concept descriptions} ${\cal EC}(\Sigma)$ is inductively defined as follows:
\begin{itemize}
\item $\bot,\top \in {\cal EC}(\Sigma)$;
\item $N_C \subseteq {\cal EC}(\Sigma)$;
\item $\forall C,D \in {\cal EC}(\Sigma), C \sqcup D, C \sqcap D, C^c \in {\cal EC}(\Sigma)$;
\item $\forall C \in {\cal EC}(\Sigma), \forall r \in N_R, \forall r.C, \exists r.C \in {\cal EC}(\Sigma)$.
\end{itemize}
\end{definition}

\begin{definition}[General concept inclusions (GCI)]
\label{atomic formulas}
Let $\Sigma$ be a signature. The set of {\bf General Concept Inclusions ($\Sigma$-GCI)} contains all the sentences of the form $C \sqsubseteq D$ and $C \equiv D$ where $C,D \in \mathcal{EC}(\Sigma)$.
\end{definition}

Sentences of the form $c \equiv C$ where $c \in N_C$ and $C \in {\cal EC}(\Sigma)$ are called {\bf concept definitions}.  $c \in N_C$ is called a {\bf defined concept} when it is defined by some sentences $c \equiv C$, and {\bf primitive concept} otherwise.

\begin{exemple}
\label{exemples formules}
The example given here is taken from Distel's PhD thesis~\cite{Dis11}. Other examples can be found in~\cite{handbook03}. \\ From the signature $\Sigma = (N_C,N_R)$ where
\begin{itemize}
\item $N_C = \{Husband,Wife,Male,Female\}$, and
\item $N_R = \{marriedTo\}$.
\end{itemize}
we can define the following GCIs:

\begin{center}
$Husband \sqsubseteq Male$ \\
$Wife \sqsubseteq Female$ \\
$Husband \equiv Male \sqcap \exists marriedTo.\top$ \\
$Wife \equiv Female \sqcap \exists marriedTo.\top$
\end{center}
\end{exemple}



\subsection{Semantics}

\begin{definition}[Model]
Let $\Sigma$ be a signature. A {\bf $\Sigma$-model} ${\cal I}$ is composed of a non-empty set (so-called {\bf carrier}) $\Delta^{\cal I}$ and a mapping $.^{\cal I}$ which associates:
\begin{itemize}
\item every concept name $c \in N_C$ with a subset $c^{\cal I} \subseteq \Delta^{\cal I}$;
\item every role name $r \in N_R$ with a binary relation $r^{\cal I} \subseteq \Delta^{\cal I} \times \Delta^{\cal I}$.
\end{itemize}
\end{definition}

\begin{exemple}
\label{exemple model}
A model $\mathcal{I}$ for the signature of Example~\ref{exemples formules} can be the following:

\begin{itemize}
\item $\Delta^\mathcal{I} = \{Marge,Homer\}$;
\item $marriedTo^\mathcal{I} = \{(Homer,Marge),(Marge,Homer)\}$;
\item $Husband^\mathcal{I} = Male^\mathcal{I} = \{Homer\}$;
\item $Wife^\mathcal{I} = Female^\mathcal{I} = \{Marge\}$.
\end{itemize}
\end{exemple}

\begin{definition}[Concept description evaluation]
\label{concept expression evaluation}
Let $\Sigma$ be a signature. Let ${\cal I}$ be a $\Sigma$-model. Let $C \in {\cal EC}(\Sigma)$ be a concept description. The {\bf evaluation} of $C$, noted $C^{\cal I}$, is inductively defined on the structure of $C$ as follows:

\begin{itemize}
\item if $C = \top$, then $C^\mathcal{I} = \Delta^\mathcal{I}$;
\item if $C = \bot$, then $C^\mathcal{I} = \emptyset$;
\item if $C = c$ with $c \in N_C$, then $C^{\cal I} = c^{\cal I}$;
\item if $C = C' \sqcup D'$ (resp. $C = C' \sqcap D'$), then $C^{\cal I} = C'^{\cal I} \cup D'^{\cal I}$ (resp. $C^{\cal I} = C'^{\cal I} \cap D'^{\cal I}$);
\item if $C = C'^c$, then $C^{\cal I} = \Delta^{\cal I} \setminus C'^{\cal I}$;
\item if $C = \forall r.C'$, then $C^{\cal I} = \{x \in \Delta^{\cal I} \mid \forall  y \in \Delta^{\cal I}, (x,y) \in r^{\cal I} \mbox{ implies } y \in C'^{\cal I}\}$;
\item if $C = \exists r.C'$, then $C^{\cal I} = \{x \in \Delta^{\cal I} \mid \exists  y \in \Delta^{\cal I}, (x,y) \in r^{\cal I} \mbox{ and } y \in C'^{\cal I}\}$.
\end{itemize}
\end{definition}

\begin{definition}[Model satisfaction]
\label{model satisfaction}
Let $\Sigma$ be a signature. Let ${\cal I}$ be a $\Sigma$-model. Let $\varphi$ be a $\Sigma$-GCI. The {\bf satisfaction} of $\varphi$ in the $\Sigma$-model ${\cal I}$, noted ${\cal I} \models \varphi$, is defined according to the form of $\varphi$ as follows:

\begin{itemize}
\item if $\varphi = (C \sqsubseteq D)$, then ${\cal I}Ê\models \varphi$ iff $C^{\cal I}Ê\subseteq D^{\cal I}$;
\item if $\varphi = (C \equiv D)$, then ${\cal I}Ê\models \varphi$ iff $C^{\cal I}Ê\subseteq D^{\cal I}$ and $D^{\cal I}Ê\subseteq C^{\cal I}$.
\end{itemize}
\end{definition}

\begin{exemple}
Obviously, the $\Sigma$-model $\mathcal{I}$ of Example~\ref{exemple model} satisfies all the formulas given in Example~\ref{exemples formules}.
\end{exemple}

\begin{definition}[Semantical consequence]
Let $\mathcal{T}$ be a set of $\Sigma$-GCIs. A $\Sigma$-GCI $\varphi$ is a {\bf semantical consequence} of $\mathcal{T}$, noted $\mathcal{T} \models \varphi$, if for every $\Sigma$-model ${\cal I}$ which satisfies every GCI in $\mathcal{T}$, ${\cal I} \models \varphi$. 
\end{definition}

\begin{exemple}
It is obvious to see that both GCIs $Husband \sqsubseteq Male$ and $Wife \sqsubseteq Female$ are semantical consequences of others formulas given in Example~\ref{exemples formules}. 
\end{exemple}

\subsection{Cyclic concept definitions}
\label{cyclic definition}

The semantics of DL given in Definitions~\ref{concept expression evaluation} and~\ref{model satisfaction} is also called {\em descriptive semantics}~\cite{Nebel91}. 
However, we can have cyclic definitions of concepts, i.e. formulas of the form $c \equiv C$ where $c$ occurs in $C$, from which it is more appropriate to interpret them with the help of fixpoint semantics. \\ Here, we restrict ourselves to simple cyclic definitions. In case of multiple cyclic definitions, i.e. a sequence of concept definitions $c_1 \equiv C_1,\ldots,c_n \equiv C_n$ such that for all $i$, $1 \leq i < n$, $c_{i+1}$ occurs in $C_i$ and $c_1$ occurs in $C_n$, it is sufficient to replace in each equation $c_i = C_i$ the defined concept $c_j$ occurring in $C_i$ by its definition $C_j$. 

\begin{exemple}
\label{example cyclic definitions}
From Example~\ref{exemples formules}, we would be able to decide that a husband is always married to a wife and a wife is always married to a husband which can be expressed by the two following equations:

\begin{center}
$Husband \equiv Male \sqcap \exists marriedTo. Wife$ \\
$Wife \equiv Female \sqcap \exists marriedTo. Husband$
\end{center}
With our restriction to consider simple cyclic definitions, this gives rise to the two following equations:

\begin{center}
$Husband \equiv Male \sqcap \exists marriedTo. (Female \sqcap \exists marriedTo. Husband)$ \\
$Wife \equiv Female \sqcap \exists marriedTo. (Male \sqcap \exists marriedTo. Wife)$
\end{center}
\end{exemple}

\medskip
To be able to get solutions to cyclic concept definitions, all the occurrences of $c$ must be within an even number of the set complementation. In this case, the defined concept name $c$ acts as a fixpoint variable, the content of which can be calculated by iterations to reach the least or the greatest fixpoint. These fixpoints are solutions of the equation $X = f^\lambda_{C}(X)$ over the complete lattice $({\cal P}(\Delta),\subseteq)$ where $\Delta$ is a domain, $\lambda : N_C \cup N_R \to {\cal P}(\Delta) \cup ({\cal P}(\Delta) \times {\cal P}(\Delta))$ is a mapping that associates with each concept name $c' \in N_C$ a subset $\lambda(c') \subseteq \Delta$ and with each role name $r \in N_R$  a binary relation $\lambda(r) \subseteq \Delta \times \Delta$, and $f^\lambda_{C} : {\cal P}(\Delta) \to {\cal P}(\Delta)$ is the mapping that maps each $X \subseteq \Delta$ to the set $Y$ inductively defined on the structure of $C$ as follows, for a cyclic concept $c$:

\begin{itemize}
\item if $C = c$, then $Y =X$;
\item if $C = c'$ with $c' \neq c \in N_C$, then $Y =\lambda(c')$; 
\item if $C = C' \sqcup D'$ (resp. $C = C' \sqcap D'$), then $Y = f^\lambda_{C'}(X) \cup f^\lambda_{D'}(X)$ (resp. $Y = f^\lambda_{C'}(X) \cap f^\lambda_{D'}(X)$);
\item if $C = C'^c$, then $Y = \Delta^{\cal I} \setminus f^\lambda_{C'}(X)$;
\item if $C = \forall r.C'$, then $Y = \{x \in \Delta \mid \forall  y \in \Delta, (x,y) \in \lambda(r) \mbox{ implies } y \in f^\lambda_{C'}(X)\}$;
\item if $C = \exists r.C'$, then $Y = \{x \in \Delta \mid \exists  y \in \Delta, (x,y) \in \lambda(r) \mbox{ and } y \in f^\lambda_{C'}(X)\}$.
\end{itemize} 
The condition which states that the occurrences of $c$ are within the scope of an even number of the set complementation $\_^c$ ensures that the mapping $f^\lambda_{C}$ is monotonous.  Tarski's fixpoint theorem~\cite{Tarski55} says that, for a monotonous function on a complete lattice, the set of fixpoints is nonempty and forms itself a complete lattice. In particular, there are a least and a greatest fixpoints. 

\medskip
Hence, let us denote a cyclic concept definition $c \equiv_{lfp} C$ (resp. $c \equiv_{gfp} C$) when we want to interpret it with the least fixpoint semantics (resp. the greatest fixpoint semantics).  \\
Given a $\Sigma$-model $\mathcal{I}$, we have: 

\begin{itemize}
\item $\mathcal{I} \models c \equiv_{lfp} C$ if and only if $c^\mathcal{I} = \bigcap\{X \in \mathcal{P}(\Delta^\mathcal{I}) \mid f^\lambda_{C}(X) \subseteq X\}$;
\item $\mathcal{I} \models c \equiv_{gfp} C$ if and only if $c^\mathcal{I} = \bigcup\{X \in \mathcal{P}(\Delta^\mathcal{I}) \mid X \subseteq f^\lambda_{C}(X)\}$  .
\end{itemize}
where $\lambda : N_C \cup N_R \to {\cal P}(\Delta) \cup ({\cal P}(\Delta) \times {\cal P}(\Delta))$ is the mapping such that $\lambda(c) = \emptyset$, for every $c' \neq c \in N_C$, $\lambda(c') = {c'}^\mathcal{I}$, and for every $r \in N_R$, $\lambda(r) = r^\mathcal{I}$. 

\begin{exemple}
The model $\mathcal{I}$ of Example~\ref{exemple model} satisfies the equation: 
\[
\begin{aligned}
&Husband  \equiv_{gfp} \\  &Male \sqcap \exists marriedTo. (Female \sqcap \exists marriedTo. Husband)
\end{aligned}
\]

\medskip
Indeed, by interpreting the above equation with the greatest fixpoint semantics, we have that $Husband^\mathcal{I} = \{Homer\}$. To show this, let us recall that the greatest fixpoint can be obtained iteratively by applying the function $f^\lambda_{C}$ to $\Delta^\mathcal{I} = \{Homer,Marge\}$ where $C = Male \sqcap \exists marriedTo. (Female \sqcap \exists marriedTo. Husband)$. After a first step, we then have that 
$$f^\lambda_{C}(\Delta^\mathcal{I})= f^\lambda_{Male}(\Delta^\mathcal{I}) \cap f^\lambda_{C'}(\Delta^\mathcal{I})$$
\noindent
where $C' = \exists marriedTo. (Female \sqcap \exists marriedTo. Husband)$. \\ By definition, we have that $f^\lambda_{Male}(\Delta^\mathcal{I}) = \{Homer\}$ and $f^\lambda_{C'}(\Delta^\mathcal{I})$ is the set 
$$\{x \in \Delta^\mathcal{I} \mid \exists y \in f^\lambda_{C''}(\Delta^\mathcal{I}), (x,y) \in marriedTo^\mathcal{I}\}$$
where $C'' = Female \sqcap \exists marriedTo. Husband$. \\ By definition, we have that $f^\lambda_{C''}(\Delta^\mathcal{I}) = f^\lambda_{Female}(\Delta^\mathcal{I}) \cap \{y \mid \exists x \in f^\lambda_{Husband}(\Delta^\mathcal{I}), (y,x) \in marriedTo^\mathcal{I}\}$. \\ But $f^\lambda_{Female}(\Delta^\mathcal{I}) = \{Marge\}$ and $f^\lambda_{Husband}(\Delta^\mathcal{I}) = \{Hormer,Marge\}$, then 
$$f^\lambda_{\exists marriedTo. Husband}(\Delta^\mathcal{I}) = \{Homer,Marge\}$$ 
\noindent
We then have that $f^\lambda_{C''}(\Delta^\mathcal{I}) = \{Marge\}$, and  $f^\lambda_{C'}(\Delta^\mathcal{I}) = \{Homer\}$, hence we have that $f^\lambda_{C}(\Delta^\mathcal{I}) = \{Homer\}$. If we iterate again, we then have that  $f ^\lambda_C(\{Homer\}) = \{Homer\}$ which is its greatest fixpoint. \\ On the contrary, interpreting the above equation with the least fixpoint semantics should yield $Husband^\mathcal{I} = \emptyset$, and then $\mathcal{I}$ does not satisfy it. 
\end{exemple}

\section{Morphism, bisimulation, preservation result and links with DL $\mathcal{ALC}$}
\label{morphism and bisimulation}

In this section, we establish some links between model morphisms in $\mathcal{ALC}$ and bisimulations. This will allow us, in particular, to generalize our result on the existence of a finite basis theorem for the DL $\mathcal{ALC}$ in Section~\ref{A finite base theorem} to complete covarieties in Section~\ref{basic constructions}.

\begin{definition}[Morphism]
\label{morphism}
Let $\Sigma$ be a signature. Let $\mathcal{I},\mathcal{I}'$ be two $\Sigma$-models. A {\bf morphism}~$\mu$ between $\Delta^\mathcal{I}$ and $\Delta^{\mathcal{I}'}$ is a mapping $\mu : \Delta^\mathcal{I} \to \Delta^{\mathcal{I}'}$ such that:

\begin{enumerate}
\item $\forall c \in N_C, \forall a \in \Delta^\mathcal{I}, a \in c^\mathcal{I} \Longleftrightarrow \mu(a) \in c^{\mathcal{I}'}$;
\item $\forall r \in N_R$,
\begin{enumerate}
\item $\mu(r^\mathcal{I}) \subseteq r^{\mathcal{I}'}$;
\item $\forall a \in \Delta^\mathcal{I}, \forall a' \in \Delta^{\mathcal{I}'}, (\mu(a),a') \in r^{\mathcal{I}'}$ \\ \mbox{} \hfill $\Rightarrow \exists b \in \Delta^\mathcal{I}, (a,b) \in r^\mathcal{I}~\mbox{and}~\mu(b) = a'$.
\end{enumerate}
\end{enumerate}
The morphism $\mu$ is a {\bf monomorphism} (resp. {\bf epimorphism}, resp. {\bf isomorphism}) if it is injective (resp. surjective, resp. bijective). 
\end{definition}

\begin{definition}[Morphism domain]
Let $\mu : \Delta^\mathcal{I} \to \Delta^{\mathcal{I}'}$ be a morphism. The $\Sigma$-model $\mathcal{I}$ is called the {\bf domain} of $\mu$. 
\end{definition}

There is a strong connection between morphisms and bisimulations. Indeed, $\Sigma$-models can be seen as coalgebras~\cite{RuttUni} with a coloring $f : \Delta^\mathcal{I} \to \mathcal{P}(N_C)$ for the functor $F_\Sigma = \mathcal{P}(\_)^{N_R} : Set \to Set$ where $Set$ is the category of sets and $\mathcal{P}$ is the powerset. Hence, given a $\Sigma$-model $\mathcal{I}$, the associated coalgebra is $(\Delta^\mathcal{I},\alpha_\mathcal{I})$ with the coloring $f^\mathcal{I}$ where:

\begin{itemize}
\item $\alpha_\mathcal{I} : \Delta^\mathcal{I} \to F_\Sigma(\Delta^\mathcal{I})$ is the mapping which associates to an individual $a \in \Delta^\mathcal{I}$ the mapping $\alpha_\mathcal{I}(a) : N_R \to \mathcal{P}(\Delta^\mathcal{I})$ such that for every $r \in N_R$, $\alpha_\mathcal{I}(a)(r) = \{b \in \Delta^\mathcal{I}\mid (a,b) \in r^\mathcal{I}\}$. 
\item $f^\mathcal{I} : \Delta^\mathcal{I} \to \mathcal{P}(N_C)$ is the mapping that associates to $a \in \Delta^\mathcal{I}$ the set $\{c~\in N_C \mid a \in c^\mathcal{I}\}$. 
\end{itemize}

Let $\mathcal{I},\mathcal{I}'$ be two $\Sigma$-models. A relation $R \subseteq \Delta^\mathcal{I} \times \Delta^{\mathcal{I}'}$ is a {\bf bisimulation} if there exists a mapping $\alpha_R : R \to F_\Sigma(R)$ such that the projections $\pi_\mathcal{I}$ and $\pi_{\mathcal{I}'}$ from $R$ to $\Delta^\mathcal{I}$ and $\Delta^{\mathcal{I}'}$ are morphisms:

$$\begin{CD}
\Delta^\mathcal{I} @<\pi_\mathcal{I}<<  R @>\pi_{\mathcal{I}'}>> \Delta^{\mathcal{I}'}\\
   @V\alpha_\mathcal{I}VV @VV\alpha_R V @VV\alpha_{\mathcal{I}'}V \\
F_\Sigma(\Delta^\mathcal{I}) @<F_\Sigma(\pi_1)<<  F_\Sigma(R) @>F_\Sigma(\pi_2)>> F_\Sigma(\Delta^{\mathcal{I}'})
\end{CD}$$
and for every $(a,a') \in R$, $f^\mathcal{I}(a) = f^{\mathcal{I}'}(a')$. 

\begin{theorem}[\cite{RuttUni}]
Let $\mathcal{I},\mathcal{I}'$ be two $\Sigma$-models. A mapping $\mu : \Delta^\mathcal{I} \to \Delta^{\mathcal{I}'}$ is a morphism if and only if its graph $G(\mu)$ is a bisimulation between $(\Delta^\mathcal{I},\alpha_\mathcal{I},f^\mathcal{I})$ and $(\Delta^{\mathcal{I}'},\alpha_{\mathcal{I}'},f^{\mathcal{I}'})$. 
\end{theorem}
Hence, morphisms are functional bisimulations. 

\medskip
It is well known that bisimulations preserve model behavior but what about GCIs? The following results answer this question.

\begin{proposition}
\label{preservation through morphisms}
Let  $\mu : \Delta^\mathcal{I} \to \Delta^{\mathcal{I}'}$ be a morphism. Then, for every $C \in \mathcal{C}(\Sigma)$, we have:
$$\forall a \in \Delta^\mathcal{I}, a \in C^\mathcal{I} \Longleftrightarrow \mu(a) \in C^{\mathcal{I}'}$$ 
\end{proposition}

\begin{proof}
By structural induction over $C$. The basic case is obvious by definition of morphism. For the induction step, several cases have to be considered:

\begin{itemize}
\item $C$ is of the form $D \sqcap E$. Let $a \in D^\mathcal{I} \cap E^\mathcal{I}$. This means that both $a \in D^\mathcal{I}$ and $a \in E^\mathcal{I}$, and then by the induction hypothesis we also have that $\mu(a) \in D^{\mathcal{I}'} \cap E^{\mathcal{I}'}$. \\ Let $\mu(a) \in D^{\mathcal{I}'} \cap E^{\mathcal{I}'}$. By the induction hypothesis, we have both $a \in D^\mathcal{I}$ and $a \in E^\mathcal{I}$, and then $a \in D^\mathcal{I} \cap E^\mathcal{I}$.
\item $C$ is of the form $D^c$. This is a direct consequence of the induction hypothesis.
\item $C$ is of the form $\exists r. D$. Let $a \in (\exists r. D)^\mathcal{I}$. This means that there exists $a' \in D^\mathcal{I}$ such that $(a,a') \in r^\mathcal{I}$. By the induction hypothesis, we have that $\mu(a') \in D^{\mathcal{I}'}$. Moreover, as $\mu$ is a morphism, we have that $(\mu(a),\mu(a')) \in r^{\mathcal{I}'}$, and then $\mu(a) \in (\exists r. D)^{\mathcal{I}'}$. \\ Let $\mu(a) \in (\exists r. D)^{\mathcal{I}'}$. This means that there exists $b' \in D^{\mathcal{I}'}$ such that $(\mu(a),b') \in r^{\mathcal{I}'}$. As $\mu$ is a morphism, there exists $a' \in \Delta^\mathcal{I}$ such that $(a,a') \in r^\mathcal{I}$ and $\mu(a') = b'$. By the induction hypothesis, we have that $a' \in D^\mathcal{I}$, and then we can conclude that $a \in (\exists r. D)^\mathcal{I}$.
\end{itemize}
\end{proof}

Actually, we have that every morphism  $\mu : \Delta^\mathcal{I} \to \Delta^{\mathcal{I}'}$ preserves GCIs from $\mathcal{I}'$ to $\mathcal{I}$. 

\begin{corollary}
\label{corollary morphism}
If there exists a morphism $\mu : \Delta^\mathcal{I} \to \Delta^{\mathcal{I}'}$, then: 
$$\mathcal{I}' \models C \sqsubseteq D \Longrightarrow \mathcal{I} \models C \sqsubseteq D$$
\end{corollary}

\begin{proof}
Let us suppose that $\mathcal{I}' \models C \sqsubseteq D$. Let $a \in C^\mathcal{I}$. By Proposition~\ref{preservation through morphisms}, we know that $\mu(a) \in C^{\mathcal{I}'}$, and then by the hypothesis, $\mu(a) \in D^{\mathcal{I}'}$, hence we can conclude by Proposition~\ref{preservation through morphisms} that $a \in D^\mathcal{I}$. 
\end{proof}

But the complete preservation of GCIs only holds for epimorphisms.

\begin{theorem}
\label{preservation through quotient}
If there exists an epimorphism $\mu : \Delta^\mathcal{I} \to \Delta^{\mathcal{I}'}$, then: 
$$\mathcal{I} \models C \sqsubseteq D \Longleftrightarrow \mathcal{I}' \models C \sqsubseteq D$$
\end{theorem}

\begin{proof}
Let us suppose that $\mathcal{I} \models C \sqsubseteq D$. Let $b \in C^{\mathcal{I}'}$. As $\mu$ is an epimorphism, there exists $a \in \Delta^\mathcal{I}$ such that $\mu(a) = b$. By Proposition~\ref{preservation through morphisms}, we have that $a \in C^\mathcal{I}$, and then by hypothesis $a \in D^\mathcal{I}$, hence by Proposition~\ref{preservation through morphisms}, we can conclude that $b \in D^{\mathcal{I}'}$. \\ The opposite direction is Corollary~\ref{corollary morphism}. 
\end{proof}

\section{A finite basis theorem for the DL $\mathcal{ALC}$}
\label{A finite base theorem}

One of the oldest questions of universal algebra was whether or not the identities of a finite algebra of a finite signature $\Sigma$ could be derived from finitely many of the identities. In universal algebras, many theorems have been obtained to positively answer this question. Here, we show a result similar to Birkhoff's theorem  which states that for every finite algebra, such a finite set of identities exists under the condition  that a finite bound is placed on the number of variables~\cite{BS81}. Here, the result we obtain will not require any condition, variables being not considered in our context. 

\begin{definition}
Let $\Sigma$ be a signature. Let $\mathcal{C}$ be a class of $\Sigma$-models. Let us note $GCI(\mathcal{C}) = \{C \sqsubseteq D \mid \forall \mathcal{I} \in \mathcal{C}, \mathcal{I} \models C \sqsubseteq D\}$. We say that $GCI(\mathcal{C}) $ is {\bf finitely based} if there is a finite set $\mathcal{T}$ of GCIs which is:

\begin{itemize}
\item {\bf Sound for $\mathcal{C}$}, i.e.  $\mathcal{T} \subseteq GCI(\mathcal{C})$; 
\item {\bf Complete for $\mathcal{C}$}, i.e. $\mathcal{T} \models  GCI(\mathcal{C})$.
\end{itemize}
\end{definition}

\begin{theorem}[Finite basis for finite model]
\label{Birkhof theorem} 
Let $\Sigma$ be a finite signature. Let $\mathcal{I}$ be a $\Sigma$-model. Then, $GCI(\{\mathcal{I}\})$ is finitely based.
\end{theorem}

\begin{proof}
Let $\Theta = \{C \equiv D \mid C^\mathcal{I} = D^\mathcal{I}\}$. By definition, $\Theta$ is an equivalence relation on $\mathcal{EC}(\Sigma)$. Moreover, as $\mathcal{I}$ is finite, there are only finitely many equivalence classes of $\Theta$ (at most $2^{\Delta^\mathcal{I}}$). For each equivalence class of $\Theta$, choose one concept. Let this set of representatives be $Q = \{C_1,\ldots,C_n\}$. Two kinds of GCIs will form the expected set $\mathcal{T}$. The first kind of GCIs consists of: $C_i,C_{i_1},C_{i_2},C_{i_3} \in Q$ with $C_{i_j} \neq C_{i_k}$ such that $j \neq k \in \{1,2,3\}$ and

\begin{itemize}
\item $c \equiv C_i$ if $c \in N_C$ and $c \equiv C_i \in \Theta$;  
\item $C_{i_1} @ C_{i_2} \equiv C_{i_3}$ if $C_{i_1} @ C_{i_2} \equiv C_{i_3} \in \Theta$ with $@ \in \{\sqcap,\sqcup\}$;
\item $C_{i_1}^c \equiv C_{i_2}$ if $C_{i_1}^c \equiv C_{i_2} \in \Theta$;
\item $Qr.C_{i_1} \equiv C_{i_2}$ if $r \in N_R$ and $Qr.C_{i_1} \equiv C_{i_2} \in \Theta$ with $Q \in \{\forall,\exists\}$.
\end{itemize}

The second kind of GCIs consists of: 
\begin{center}
$C_{i_1} \sqsubseteq C_{i_2}$ if $C^\mathcal{I}_{i_1} \subseteq C^\mathcal{I}_{i_2}$
\end{center}

with $C_{i_1},C_{i_2} \in Q$ and $C_{i_1} \neq C_{i_2}$.

\medskip
Soundness is obvious by construction. Then, let us show the completeness. First, let us show by structural induction over $C \in \mathcal{EC}(\Sigma)$ that if $C \equiv C_i \in \Theta$, then $\mathcal{T} \models C \equiv C_i$ for $C_i \in Q$. 

\begin{itemize}
\item {\em Basic case.} This is obvious by definition.
\item {\em General case.} Several cases have to be considered:

\begin{itemize}
\item Let $C = D_1 @ D_2$ with $@ \in \{\sqcap,\sqcup\}$. By construction, there exists $C_{i_1},C_{i_2} \in Q$ such that $D_1 \equiv C_{i_1},D_2 \equiv C_{i_2} \in \Theta$. By the induction hypothesis, we then have that $\mathcal{T} \models D_j \equiv C_{i_j}$ for $j = 1,2$. Here, two cases have to be considered:
\begin{enumerate}
\item $i_1 = i_2$. In this case, we have that $\mathcal{T} \models D_1 \equiv D_2$, and then $\mathcal{T} \models D_j \equiv C_i$ for $j = 1,2$. We can then conclude that $\mathcal{T} \models C \equiv C_i$. 
\item  $i_1 \neq i_2$. Hence, $C_{i_1} @ C_{i_2} \equiv C_i \in \Theta$, then $C_{i_1} @ C_{i_2} \equiv C_i \in \mathcal{T}$, and we can conclude that $\mathcal{T} \models C \equiv C_i$. 
\end{enumerate} 
\item Let $C = D^c$. By construction, there exists $C_{i_1} \in Q$ such that $D \equiv C_{i_1} \in \Theta$. By the induction hypothesis, we then have that $\mathcal{T} \models D \equiv C_{i_1}$, and then $\mathcal{T} \models D^c \equiv C_{i_1}^c$. Hence, $C_{i_1}^c \equiv C_i \in \Theta$, then $C_{i_1}^c \equiv C_i \in \mathcal{T}$, and we can conclude that $\mathcal{T} \models C \equiv C_i$.
\item Let $C = Qr.D$ with $Q \in \{\forall,\exists\}$. By construction, there exists $C_{i_1} \in Q$ such that $D \equiv C_{i_1} \in \Theta$. By induction hypothesis, we then have that $\mathcal{T} \models D \equiv C_{i_1}$. Hence, $Qr.D \equiv Qr.C_{i_1} \in \Theta$, then $Qr.C_{i_1} \equiv C_i \in \mathcal{T}$, and we can conclude that $\mathcal{T} \models C \equiv C_i$.
\end{itemize}
\end{itemize}
Hence, if $\mathcal{I} \models C \equiv D$, then $\mathcal{T} \models C \equiv D$. Let us suppose that $(C \sqsubseteq D) \in GCI(\{\mathcal{I}\})$ such that $C \equiv D \not\in \Theta$. By definition, there are $C_i,C_j \in Q$ such that $(C \equiv C_i),(D\equiv C_j) \in \Theta$, and then $\mathcal{T} \models C \equiv C_i$ and $\mathcal{T} \models D \equiv C_j$. By the hypothesis that $(C \sqsubseteq D) \in GCI(\{\mathcal{I}\})$ such that $(C \equiv D) \not\in \Theta$, we then have that $C_i \sqsubseteq C_j \in \mathcal{T}$. We can then conclude that $\mathcal{T} \models C \sqsubseteq D$. 
\end{proof}

The set of axioms $\mathcal{T}$ obtained by the algorithm given in the proof of Theorem~\ref{Birkhof theorem} is unlikely to be minimal, i.e. such that no strict subset of $\mathcal{T}$ is complete for $\mathcal{I}$ (see Example~\ref{example of finite basis}). In fact, we may generate a lot of tautologies and GCIs that can be inferred from others. Therefore, an elimination step is still needed to keep only the GCIS which have an axiom status (i.e. GCIs which are not tautologies and cannot be inferred from others). This tedious work can be automated, reasoning in the logic $\mathcal{ALC}$ being computable. 


\section{An algorithm for computing a finite basis for the DL $\mathcal{ALC}$}
\label{algorithm}

With a simple condition on the model $\mathcal{I}$, we can effectively define for each subset $S \in 2^{\Delta^\mathcal{I}}$ a representative of the equivalence class $\Gamma \in \Theta$ such that for every $C \in \Gamma$, $C^\mathcal{I} = S$. For this, we need the following notations:

\begin{itemize}
\item Let $r \in N_R$ be a relation name. Let us note $r^\mathcal{I}_1 = \{a \in \Delta^\mathcal{I} \mid \exists b \in \Delta^\mathcal{I}, (a,b) \in r^\mathcal{I}\}$.
\item Let $a \in \Delta^\mathcal{I}$. Let us note $N_C^\mathcal{I}(a) = \{c \in N_C \mid a \in c^\mathcal{I}\}$. 
\end{itemize}

For our algorithm, let us suppose that the following condition is satisfied in $\mathcal{I}$: $\forall a,b \in \Delta^\mathcal{I}$,
$$N_C^\mathcal{I}(a) = N_C^\mathcal{I}(b) \Longrightarrow \exists r \in N_R, a \in r^\mathcal{I}_1 \mbox{ and } b \not\in  r^\mathcal{I}_1$$

\begin{exemple}
The model $\mathcal{I}$ of Example~\ref{exemple model} satisfies this condition. 
\end{exemple}

When dealing with finite models, this condition can be effectively checked.

\medskip
Let us define recursively on the cardinality of the subsets $S \in 2^{\Delta^\mathcal{I}}$ the concept $C_S$ such that $C_S^\mathcal{I} = S$ as follows:

\begin{itemize}
\item For $S = \emptyset$, $C_S = \bot$;
\item For $S = \{a\}$, we apply the following sequence of actions:

\begin{enumerate}
\item if $N_C^\mathcal{I}(a) \neq \emptyset$, then let us choose $c \in N_C^\mathcal{I}(a)$ and let us set $C_a = c$. Otherwise, let us choose any $c \in N_C$ and let us set $C_a = c^c$;
\item for every $b \neq a \in \Delta^\mathcal{I}$, let us define $C_{b/a}$ as follows: 
\begin{itemize}
\item if $N_C^\mathcal{I}(a) \neq N_C^\mathcal{I}(b)$, then let us choose $c \in N_C^\mathcal{I}(b) \setminus N_C^\mathcal{I}(a)$ and let us set $C_{b/a} = c$;
\item otherwise (i.e. $N_C^\mathcal{I}(a) = N_C^\mathcal{I}(b)$), let us choose $r \in N_R$ such that $b \in r^\mathcal{I}_1$ and $a \not\in r^\mathcal{I}_1$ and let us set $C_{b/a} = \exists r. \top$.
\end{itemize}
\end{enumerate}
And then let us set $C_{\{a\}} = C_a \sqcap \bigsqcap_{b \neq a \in C^\mathcal{I}_a}C^c_{b/a}$.
\item For $S = \{a_1,\ldots,a_n\}$, $C_S = \bigsqcup_{a_i \in S} C_{\{a_i\}}$.
\end{itemize}
Hence, we have a representative for each equivalence class, and then by applying the rules given in the proof of Theorem~\ref{Birkhof theorem}, we can generate automatically a complete finite base for the model $\mathcal{I}$ under consideration.  

\begin{exemple}
\label{example of finite basis}
For the model $\mathcal{I}$ of Example~\ref{exemple model}, the algorithm can yield for the sets $\emptyset$, $\{Homer\}$, $\{Marge\}$, and $\{Homer,Marge\}$, the basic concepts $\bot$, $Male$, $Wife$ and $Male \sqcup Wife$. According to the algorithm given in the proof of Theorem~\ref{Birkhof theorem}, this gives rise to the following sets of GCIs:

\begin{center}
$Husband \equiv Male$ \\
$Female \equiv Wife$ \\
$Male \sqcap Wife \equiv \bot$ \\
$Male \sqcap (Male \sqcup Wife) \equiv Male$ \\
$Wife \sqcap (Male \sqcup Wife) \equiv Wife$ \\
$Male \sqcup (Male \sqcup Wife) \equiv Male \sqcup Wife$ \\
$Wife \sqcup (Male \sqcup Wife) \equiv Male \sqcup Wife$ \\
$\bot \sqcup C \equiv C$ with $C \in \{Male,Wife,Male \sqcup Wife\}$ \\
$\bot \sqcap C \equiv \bot$ with $C \in \{Male,Wife,Male \sqcup Wife\}$ \\
$\bot^c = Male \sqcup Wife$ \\
$Male^c \equiv Wife$ \\
$Wife^c \equiv Male$ \\
$(Male \sqcup Wife)^c \equiv \bot$ \\
$@ marriedTo. Wife \equiv Male$ \\with $@ \in \{\exists,\forall\}$ \\
$@ marriedTo. Male \equiv Wife$ \\with $@ \in \{\exists,\forall\}$ \\
$@ marriedTo.C \equiv C$ with $@ \in \{\exists,\forall\}$ and $C \in \{\bot,Male \sqcup Wife\}$ \\
$\bot \sqsubseteq C$ with $C \in \{Male,Wife,Male \sqcup Wife\}$ \\
$Male \sqsubseteq Male \sqcup Wife$ \\
$Wife \sqsubseteq Male \sqcup Wife$
\end{center} 
As noted above, this set of axioms is not minimal. An elimination step has then to be performed. This gives rise to the following set:

\begin{center}
$Husband \equiv Male$ \\
$Female \equiv Wife$ \\
$Male \sqcap Wife \equiv \bot$ \\
$Male^c \equiv Wife$ \\
$(Male \sqcup Wife)^c \equiv \bot$ \\
$@ marriedTo. Wife \equiv Male$ \\with $@ \in \{\exists,\forall\}$ \\
$@ marriedTo. Male \equiv Wife$ \\with $@ \in \{\exists,\forall\}$
\end{center} 
\end{exemple}

\begin{theorem}[Correctness]
Let $\mathcal{I}$ be a finite $\Sigma$-model. Then, for every subset $S \subseteq \Delta^\mathcal{I}$, the concept $C_S$ calculated by the procedure above satisfies $C^\mathcal{I}_S = S$. 
\end{theorem}

\begin{proof}
The only difficulty of the proof is when $S$ is a singleton, say $\{a\}$. The rest of the proof is straightforward. Let us then show that $C^\mathcal{I}_{\{a\}} = \{a\}$. By Point (i) of the procedure, we start by defining a first concept $C_a$ which is either a basic concept $c$ or the set difference of a basic concept $c^c$ depending on whether $N^\mathcal{I}_C(a)$ is non-empty or not. By definition of $C_a$, we obviously have that $a \in C^\mathcal{I}_a$. Let $b \neq a \in C^\mathcal{I}_a$. Here, two cases have to be considered:
\begin{enumerate}
\item $N^\mathcal{I}_C(a) \neq N^\mathcal{I}_C(b)$. The procedure then sets $C_{b/a} = c$. By definition of $C_{b/a}$, we have both $a \not\in C^\mathcal{I}_{b/a}$ and $b \in C^\mathcal{I}_{b/a}$. Hence, we have that $b \not\in (C_a \sqcap C^c_{b/a})^\mathcal{I}$ and $a \in (C_a \sqcap C^c_{b/a})^\mathcal{I}$.
\item $N^\mathcal{I}_C(a) = N^\mathcal{I}_C(b)$. Let $r$ be a relation in $N_R$ such that $b \in r^\mathcal{I}_1$ but $a \not\in r^\mathcal{I}_1$. Such a relation $r$ exists by hypothesis. The procedure sets $C_{b/a} = \exists r. \top$, and then we have $b \in C^\mathcal{I}_{b/a}$ and $a \not\in C^\mathcal{I}_{b/a}$, hence we conclude that $b \not\in (C_a \sqcap  C^c_{b/a})^\mathcal{I}$. 
\end{enumerate} 
As $(C_a \sqcap  C^c_{b/a})^\mathcal{I} \subseteq C^\mathcal{I}_a$, the procedure at each step (ii) removes an element $b \neq a$ of $C_a$. As $C^\mathcal{I}_a$ is finite, in a finite number of steps, the procedure generates the concept $C_{\{a\}}$ which satisfies by construction $C^\mathcal{I}_{\{a\}} = \{a\}$.
\end{proof}

\section{Extension to complete covarieties}
\label{basic constructions}

This section deals with the extension of the fundamental theorem on the existence of a finite basis theorem for the DL $\mathcal{ALC}$ to complete covarieties.
Complete covarieties are classes of $\Sigma$-models which are closed under morphism domain, quotients (homomorphic images) and coproducts.

\begin{definition}[Homomorphic image]
Let $\mathcal{I}$ and $\mathcal{I}'$ be two $\Sigma$-models. $\mathcal{I}'$ is a {\bf homomorphic image} of $\mathcal{I}$ if there exists an epimorphism $\mu :\Delta^\mathcal{I} \to \Delta^{\mathcal{I}'}$.  
\end{definition}


\begin{definition}[Coproduct]
Let $(\mathcal{I}_i)_{i \in \Lambda}$ be a family of $\Sigma$-models indexed by a set $\Lambda$. Let $\sum_{i \in \Lambda} \mathcal{I}_i$ denote the $\Sigma$-model $\mathcal{I}'$ defined by:

\begin{itemize}
\item $\Delta^{\mathcal{I}'} = \{(i,a) \mid i \in \Lambda, a \in \Delta^{\mathcal{I}_i}\}$;
\item $\forall c \in N_C, c^{\mathcal{I}'} = \bigcup_{i \in \Lambda} \{(i,a) \mid a \in c^{\mathcal{I}_i}\}$;
\item $\forall r \in N_R, r^{\mathcal{I}'} = \bigcup_{i \in \Lambda} \{((i,a),(i,b)) \mid (a,b) \in r^{\mathcal{I}_i}\}$.
\end{itemize} 
$\sum_{i \in \Lambda} \mathcal{I}_i$ is called {\bf coproduct} of  $(\mathcal{I}_i)_{i \in \Lambda}$. 
\end{definition}

Complete covarieties are trivially covarieties because morphism domains contain embeddings, that is, given a model $\mathcal{I}'$, all the models $\mathcal{I}$ such that there exists a monomorphism $\mu : \Delta^\mathcal{I} \to \Delta^{\mathcal{I}'}$. 

\medskip
We saw in Section~\ref{morphism and bisimulation} how GCIs are preserved through morphisms and epimorphisms, and then through morphism domain and homomorphic image closures. Here, we show how GCIs are preserved through coproduct closure.

\begin{proposition}
\label{preservation coproduct}
Let $\mathcal{I}'$ be the coproduct of $(\mathcal{I}_i)_{i \in \Lambda}$. Then, for every $C \in \mathcal{C}(\Sigma)$, we have that $C^{\mathcal{I}'} = \bigcup_{i \in \Lambda} \{i\} \times C^{\mathcal{I}_i}$. 
\end{proposition}

\begin{proof}
By induction on $C$. The basic case is obvious by definition. For the induction step, several cases have to be considered:

\begin{itemize}
\item $C$ is of the form $D \sqcap E$. By the induction hypothesis, we have that $D^{\mathcal{I}'} = \bigcup_{i \in \Lambda} \{i\} \times D^{\mathcal{I}_i}$ and $E^{\mathcal{I}'} = \bigcup_{i \in \Lambda} \{i\} \times E^{\mathcal{I}_i}$. By developing $\bigcup_{i \in \Lambda} \{i\} \times D^{\mathcal{I}_i} \cap \bigcup_{i \in \Lambda} \{i\} \times E^{\mathcal{I}_i}$, we obtain $\bigcup_{i \in \Lambda} \bigcup_{j \in \Lambda} ( \{i\} \times D^{\mathcal{I}_i} \cap \{j\} \times D^{\mathcal{I}_j})$. By definition of coproduct, for every $i,j$ such that $i \neq j$, we have that $\{i\} \times D^{\mathcal{I}_i} \cap \{j\} \times E^{\mathcal{I}_j} = \emptyset$, and then we can conclude that $C ^{\mathcal{I}'} = \bigcup_{i \in \Lambda} (\{i\} \times D^{\mathcal{I}_i} \cap \{i\} \times E^{\mathcal{I}_i})$. 
\item $C$ is of the form $D^c$. By the induction hypothesis, we have that $D^{\mathcal{I}'} = \bigcup_{i \in \Lambda} \{i\} \times D^{\mathcal{I}_i}$. Let us show that $\Delta^{\mathcal{I}'} \setminus D^{\mathcal{I}'} = \bigcup_{i \in \Lambda} \{i\} \times \Delta^{\mathcal{I}_i} \setminus D^{\mathcal{I}_i}$. Let $(i,a) \in \Delta^{\mathcal{I}'} \setminus D^{\mathcal{I}'}$. By the induction hypothesis, this means that $(i,a) \not\in D^{\mathcal{I}'}$, and then by the definition of coproduct, $(i,a) \not\in \{i\} \times \Delta^{\mathcal{I}_i} \setminus D^{\mathcal{I}_i}$. \\
Let $(i,a) \in \bigcup_{i \in \Lambda} \{i\} \times \Delta^{\mathcal{I}_i} \setminus D^{\mathcal{I}_i}$. This means that $(i,a) \not\in \{i\} \times D^{\mathcal{I}_i}$, and then by the definition of coproduct $(i,a) \in \bigcup_{i \in \Lambda} \{i\} \times \Delta^{\mathcal{I}_i} \setminus D^{\mathcal{I}_i}$. By the induction hypothesis, we have then $(i,a) \not\in D^{\mathcal{I}'}$, hence we can conclude that $(i,a) \in \Delta^{\mathcal{I}'} \setminus D^{\mathcal{I}'}$. 
\item $C$ is of the form $\exists r. D$. By the induction hypothesis, we have that $D^{\mathcal{I}'} = \bigcup_{i \in \Lambda} \{i\} \times D^{\mathcal{I}_i}$. By the definition of coproduct, we have that:
$$\begin{array}{ll}
\exists r.D^{\mathcal{I}'} & = \{(i,a) \mid \exists b \in D^{\mathcal{I}_i}, (a,b) \in r^{\mathcal{I}_i}\} \\
                                      & = \bigcup_{i \in \Lambda} \{i\} \times \{a \mid \exists b \in D^{\mathcal{I}_i},(a,b) \in r^{\mathcal{I}_i}\} \\
                                      & = \bigcup_{i \in \Lambda} \{i\} \times (\exists r.D)^{\mathcal{I}_i}
\end{array}$$
\item The cases where $C$ is of the form $\forall r.D$ or $D\sqcup E$ are directly derived from the previous ones since $(\forall r. D)^c = \exists r. D^c$ and $(C \sqcap D)^c = C^c \sqcup D^c$.
\end{itemize} 
\end{proof}

\begin{corollary}
\label{corollary coproduct}
Let $\mathcal{I}'$ be the coproduct of $(\mathcal{I}_i)_{i \in \Lambda}$. Then: $(\forall i \in \Lambda, \mathcal{I}_i \models C \sqsubseteq D) \Longleftrightarrow \mathcal{I}' \models C \sqsubseteq D$.
\end{corollary}

\begin{proof}
The result is direct from Proposition~\ref{preservation coproduct}. 
\end{proof}

\begin{definition}[Complete covariety]
A class of $\Sigma$-models $\mathcal{C}v$ is a {\bf complete covariety} if it is closed under morphism domain, homomorphic images, and coproducts. \\ Let $\mathcal{K}$ be a class of $\Sigma$-models. Let $\mathcal{C}v(\mathcal{K})$ denote the complete covariety generated by $\mathcal{K}$. \\ A complete covariety $\mathcal{C}v$ is {\bf finitely generated} if $\mathcal{C}v = \mathcal{C}v(\mathcal{K})$ for some finite set of finite $\Sigma$-models $\mathcal{K}$. 
\end{definition}

Let us now introduce the notions of final, weakly final, and behavioral models to prove,
when there exists a final model $\mathbb{T}$, 
an extension for complete covarieties of Theorem 15.3 given by J. Rutten in~\cite{RuttUni}, that states that any complete covariety is determined by a submodel of $\mathbb{T}$. 


\begin{definition}[Final and weakly final models]
A model is said to be \textbf{final} if for every $\Sigma$-model $\mathcal{I}$ there is a unique morphism  $\mu_\mathcal{I} : \Delta^\mathcal{I} \to \Delta^\mathbb{T}$ (i.e. there exists a unique morphism from any model to it), and \textbf{weakly final} if the morphism is not unique.
\end{definition}

\begin{definition}[Behavioral model]
Let $\Sigma$ be a signature. Let us note $N_R^\infty$ the set of infinite and finite words on $N_R$. Let us define the $\Sigma$-model $\mathbb{T}$, called {\bf behavioral model}, as follows:

\begin{itemize}
\item $\Delta^\mathbb{T} = \mathcal{P}(N^\infty_R) \times \mathcal{P}(N_C)$;
\item for every $c \in N_C$, $c^\mathbb{T} = \{(R,C) \mid c \in C~\text{and } R \in N^\infty_R\}$;
\item for every $r \in N_R$, $((R_1,C_1),(R_2,C_2)) \in r^\mathbb{T}$ iff  $R_2$ is the set $\{r_1 \cdot r_2 \ldots r_n \ldots \mid r \cdot r_1 \cdot r_2 \ldots r_n \ldots \in R_1\}$.
\end{itemize}
\end{definition}

%
\begin{theorem}
\label{weakly final}
Every behavioral model $\mathbb{T}$ is weakly final.
\end{theorem}

\begin{proof}
Let $\mathcal{I}$ be a $\Sigma$-model. For every $a \in \Delta^\mathcal{I}$, let us note $beh(a) \subseteq \mathcal{P}(N^\infty_R)$ the set defined by:

$$r_1 \ldots r_n \ldots \in beh(a) \Leftrightarrow 
\left\{ \begin{array}{l}
\exists a_o,a_1,\ldots,a_n,\ldots \in \Delta^\mathcal{I}  \text{ such that} \\
\forall i \in \mathbb{N}, (a_i,a_{i+1}) \in r_{i+1} \text{ and} \\
a_0 = a
\end{array} \right.$$

Let us define the mapping $\mu_\mathcal{I}  : a \in \Delta^\mathcal{I} \mapsto (beh(a),\{c \in N_C \mid a \in c^\mathcal{I}\})$. By definition, we have for every $c \in N_C$ and every $a \in \Delta^\mathcal{I}$ that $a \in c^\mathcal{I} \Leftrightarrow \mu(a) \in c^\mathbb{T}$. \\ In the same way, it is not difficult from the definitions of morphism and the behavioral model to show that, for every $r \in N_R$, the two conditions of Definition~\ref{morphism} are satisfied by the mapping $\mu_\mathcal{I} $.
\end{proof}
It is well known that to have a unique morphism $\mu_\mathcal{I} $ between $\mathcal{I}$ and $\mathbb{T}$ some restrictions have to be imposed on the cardinality of the first set $R$ of any element $(R,C) \in \Delta^\mathbb{T}$ (see~\cite{RuttUni}). The reason is that final models are isomorphic, and then by using the notations of coalgebras, we would have that $\mathbb{T} \cong F_\Sigma(\mathbb{T})$ which is a contradiction because for any set $S$ the cardinality of $N_R \times \mathcal{P}(S)$ is greater that of $S$. Therefore, if we restrict the functor $F_\Sigma$ to the functor $F'_\Sigma = \mathcal{P}_{\leq \kappa}(\_)^{N_R} : Set \to Set$ for a given cardinality $\kappa$ where $\mathcal{P}_{\leq \kappa}(S) = \{U \mid U \subseteq S~\mbox{and}~|U| \leq \kappa\}$, then $\mu_\mathcal{I} $ such as defined in the proof of Theorem~\ref{weakly final} is unique and then $\mathbb{T}$ is final\footnote{In \cite{RuttUni}, $\mathbb{T}$ when it is final, it is also said {\bf cofree} on $\mathcal{P}(N_C)$.}.  

\begin{theorem}[Characterization]
\label{characterization}
For any complete covariety $\mathcal{C}v$, there exists a submodel $\mathcal{U}$ of $\mathbb{T}$ such that $\mathcal{C}v = \mathcal{C}v(\mathcal{U})$. 
\end{theorem}

\begin{proof}
Let $\mathcal{C}v$ be a complete covariety. Let us define the model $\mathcal{U}$ as:
$$\Delta^\mathcal{U} = \bigcup \{\mu_\mathcal{I}(\Delta^\mathcal{I}) \mid \mathcal{I} \in \mathcal{C}v\}$$
where $\mu_{\mathcal{I}}: \Delta^\mathcal{I} \to \Delta^\mathbb{T}$ is any morphism defined in Theorem~\ref{weakly final}. 

Obviously, each $\mu_\mathcal{I}(\Delta^\mathcal{I})$ is embedded into $\Delta^\mathcal{U}$, i.e. $\mu_\mathcal{I} : \Delta^\mathcal{I} \to \Delta^\mathcal{U}$ is a monomorphism. Finally, the union of embeddings is again an embedding of $\mathcal{U}$. This allows us to conclude that $\mathcal{C}v \subseteq \mathcal{C}v(\mathcal{U})$. \\ For the converse, let us prove first that $\mathcal{U} \in \mathcal{C}v$. For every $i \in \Delta^\mathcal{U}$, let us choose a model $\mathcal{I}_i \in \mathcal{C}v$ such that $i \in \mu_{\mathcal{I}_i}(\mathcal{I}_i)$. Obviously, we have an epimorphism $q : \sum_{i \in \Delta^\mathcal{U}} \mathcal{I}_i \to \mathcal{U}$ which allows us to conclude that $\mathcal{U} \in \mathcal{C}v$. Now, every model $\mathcal{I} \in \mathcal{C}v(\mathcal{U})$ is obtained by application of morphism domain, homomorphic image and coproduct operators from $\mathcal{U}$. By induction over the way $\mathcal{I}$ has been obtained, we can easily show that $\mathcal{I} \in \mathcal{C}v$. 
\end{proof}

We have the following result which gives one direction of the dual of Birkhoff's variety theorem for complete covarieties of $\Sigma$-models\footnote{Birkhoff's variety theorem for algebras~\cite{BS81} states that any class of algebras is closed under the formation of subalegbras, homomorphic images and product if and only if it is equationally definable.} with respect to GCIs. Let $\mathbb{M}$ be a class of $\Sigma$-models. We say that $\mathbb{M}$ is a {\bf GCI class} if there exists a set $\mathcal{T}$ of $\Sigma$-GCIs such that for every $\Sigma$-model $\mathcal{I}$, we have:
$$(\forall C \sqsubseteq D \in \mathcal{T}, \mathcal{I} \models C \sqsubseteq D) \Longleftrightarrow \mathcal{I} \in \mathbb{M}$$ 

\begin{theorem}
Let $\mathbb{M}$ be a GCI class. Then, $\mathbb{M}$ is a complete covariety.
\end{theorem}

\begin{proof}
Let $\mathcal{T}$ be the set of GCIs satisfied by all the models in $\mathbb{M}$. Therefore, by Theorem~\ref{preservation through quotient}, and Corollaries~\ref{corollary morphism} and \ref{corollary coproduct} we have for every model $\mathcal{I} \in Cv(\mathbb{M})$ that  $\mathcal{I} \models \mathcal{T}$ and then $\mathcal{C}v(\mathbb{M}) \subseteq \mathbb{M}$. As we obviously have that $\mathbb{M} \subseteq \mathcal{C}v(\mathbb{M})$, then we can conclude that $\mathcal{C}v(\mathbb{M}) = \mathbb{M}$.
\end{proof}

Unfortunately, the opposite direction, which would state that any complete covariety is a GCI class, fails. Conventionally, this result requires that the class of models for a given signature has final models. However, this is not enough. In fact, to have the property that complete covarieties are GCI classes, we should be able to express formulas that describe behaviors, i.e. formulas that ensure the existence of paths of the form $r_1 \ldots r_n \ldots$. Indeed, given a complete covariety $\mathcal{C}v$, if we denote $GCI(\mathcal{C}v) = \{C \sqsubseteq D \mid \forall \mathcal{I} \in \mathcal{C}v, \mathcal{I} \models C \sqsubseteq D\}$ and $\mathbb{M} = \{\mathcal{I} \mid \mathcal{I} \models GCI(\mathcal{C}v)\}$, then clearly, by Theorem~\ref{preservation through quotient}, and Corollaries~\ref{corollary morphism} and \ref{corollary coproduct}, $\mathbb{M}$ is a complete covariety. Moreover, we have that $\mathcal{C}v \subseteq \mathbb{M}$ and $GCI(\mathbb{M}) = GCI(\mathcal{C}v)$. \\ 
By Theorem~\ref{characterization}, we know that there exist $\mathcal{U}$ and $\mathcal{U}'$ submodels of the final model $\mathbb{T}$ such that $\mathcal{C}v = \mathcal{C}v(\mathcal{U})$ and $\mathbb{M} = \mathcal{C}v(\mathcal{U}')$. We then have that  $\mathcal{U}$ is a submodel of $\mathcal{U}'$ and: $\forall C \sqsubseteq D, \mathcal{U} \models C \sqsubseteq D \Leftrightarrow \mathcal{U}' \models C \sqsubseteq D$. However, as noted above, $\mathcal{U}'$ may not be a submodel of $\mathcal{U}$ (which would lead to $\mathcal{U}' = \mathcal{U}$ and then $\mathcal{C}v = \mathbb{M}$) because to ensure that, we should be able to express properties about model behavior. But such formulas on behavior are not expressible by GCIs, and then the logic $\cal{ALC}$ is not expressive enough according to the definition given by A. Kurz in~\cite{Kurz00}.

\medskip
We can now simply extend  Theorem~\ref{Birkhof theorem} to finitely generated complete covariety.

\begin{theorem}
Let $\Sigma$ be a signature. Let $\mathcal{C}v$ be a finitely generated complete covariety over $\Sigma$. Then, $GCI(Cv)$ is finitely based.
\end{theorem} 

\begin{proof}
Let $\mathcal{K}$ be the finite set of finite $\Sigma$-models such that $\mathcal{C}v = \mathcal{C}v(\mathcal{K})$. It is easy to show that $\mathcal{C}v(\mathcal{K}) = \mathcal{C}v(\sum_{\mathcal{I} \in \mathcal{K}} \mathcal{I})$. By Theorem~\ref{preservation through quotient}, and Corollaries~\ref{corollary morphism} and \ref{corollary coproduct}, we have that $GCI(\mathcal{C}v) = GCI(\{\sum_{\mathcal{I} \in \mathcal{K}} \mathcal{I}\})$. But, by Theorem~\ref{Birkhof theorem}, we know that $GCI(\{\sum_{\mathcal{I} \in \mathcal{K}} \mathcal{I}\})$ is finitely based, then so is $GCI(\mathcal{C}v)$. 
\end{proof}

\section{Conclusion}
The aim of this paper was threefold: (i) proving the existence of a finite basis for the description logic $\mathcal{ALC}$, (ii) studying the conditions to effectively build this finite basis, and introducing a concrete algorithm to do so, and (iii) extending this result to complete covarieties. 

Characterizing and building a finite basis for a given description is of prime importance in several ontology-related applications, such as learning terminologies or  combining formal concept analysis and description logics, which paves the way for further non-classical reasoning services (e.g. axiom pinpointing, etc.). Generalizing our result to complete covarieties was motivated by building a bridge between the work on description logics and abstract algebra, thus enlarging its scope to other knowledge representation formalisms. Future work will deal with the implementation of the algorithm in $\mathcal{ALC}$, and studying the impact of the theorem on complete covarieties to other logical formalisms.

\bibliographystyle{elsarticle-num}
\bibliography{biblio}

\begin{thebibliography}{10}
\expandafter\ifx\csname url\endcsname\relax
  \def\url#1{\texttt{#1}}\fi
\expandafter\ifx\csname urlprefix\endcsname\relax\def\urlprefix{URL }\fi
\expandafter\ifx\csname href\endcsname\relax
  \def\href#1#2{#2} \def\path#1{#1}\fi

\bibitem{handbook03}
F.~Baader, D.~Calvanese, D.~McGuinness, D.~Nardi, P.-F. Patel-Schneider (Eds.),
  The Description Logic Handbook: Theory, Implementation, and Applications.,
  Cambridge University Press, 2003.

\bibitem{Min81}
M.~Minsky, A framework for representing knowledge., in: J.~Haudeland (Ed.),
  Mind Design: Philosophy, Psychology, Artificial Intelligence., The MIT Press,
  1981.

\bibitem{Sow91}
J.-F. Sowa, Principles of Semantic Networks, Morgan Kaufmann, 1991.

\bibitem{BBL05}
F.~Baader, S.~Brandt, C.~Lutz, Pushing the ${\cal el}$ envelop, in: {IJCAI}
  2005, Proceedings of the 19th International Joint Conference on Artificial
  Intelligence, 2005, pp. 364--369.

\bibitem{BLS06}
F.~Baader, C.~Lutz, B.~Suntisrivaraporn, {CEL}-a polynomial-time reasoner for
  life cycle ontologies, in: {IJCAR} 2006, Proceedings of the 20th
  International Joint Conference on Artificial Intelligence, Vol. 4130 of
  Lecture Notes in Artificial Intelligence, Springer, 2006, pp. 287--291.

\bibitem{Dis11}
F.~Distel, Learning description logic knowledge bases from data using methods
  from formal concept analysis, Ph.D. thesis, University of Dresden (2011).

\bibitem{BD08}
F.~Baader, F.~Distel, A finite basis for the set of el-implications holding in
  a finite model, in: R.~Medina, S.-A. Obiedkov (Eds.), Formal Concept
  Analysis, 6th International Conference, {ICFCA} 2008, Vol. 4933 of Lecture
  Notes in Computer Science, Springer, 2008, pp. 46--61.

\bibitem{BD09}
F.~Baader, F.~Distel, Exploring finite models in the description logic, in:
  S.~Ferr{\'{e}}, S.~Rudolph (Eds.), Formal Concept Analysis, 7th International
  Conference, {ICFCA} 2009, Vol. 5548 of Lecture Notes in Computer Science,
  Springer, 2009, pp. 146--161.

\bibitem{GW97}
B.~Ganter, R.~Wille, {F}ormal {C}oncept {A}nalysis: {M}athematical
  {F}oundations, Springer-Verlag, 1997.

\bibitem{BS81}
S.~Burris, H.-P. Sankappanavar, A course in Universal Algebra, Graduate Texts
  in Mathematics, Springer-Verlag, 1981.

\bibitem{GS01}
H.-P. Gumm, T.~Schroder, Covarieties and complete covarieties, Theoretical
  Computer Science 260 (2001) 71--86.

\bibitem{BHS08}
F.~Baader, I.~Horrocks, U.~Sattler, Handbook of Knowledge Representation,
  Elsevier, 2008, Ch. {D}escription {L}ogics.

\bibitem{AN01}
A.~Arnold, D.~Niwinski, The $\mu$-calculus over power set algebras, Elsevier,
  2001, Ch. Rudiments of $\mu$-calculus, pp. 141--153.

\bibitem{Koz83}
D.~Kozen, Results on the propositional $\mu$-calculus, Theoretical Computer
  Science 27~(3) (1983) 333--354.

\bibitem{Schild94}
K.~Schild, Terminological cycles and the propositional $\mu$-calculus, in:
  {KR}'94, 4th Int. Conf. on the Principles of Knowledge Representation and
  Reasoning, 1994, pp. 509--520.

\bibitem{RuttUni}
J.~J. M.~M. Rutten,
  \href{http://dx.doi.org/10.1016/S0304-3975(00)00056-6}{Universal coalgebra: a
  theory of systems}, Theoretical Computer Science 249~(1) (2000) 3--80.
\newblock \href {http://dx.doi.org/10.1016/S0304-3975(00)00056-6}
  {\path{doi:10.1016/S0304-3975(00)00056-6}}.
\newline\urlprefix\url{http://dx.doi.org/10.1016/S0304-3975(00)00056-6}

\bibitem{Nebel91}
B.~Nebel, Terminological cycles: {S}emantics and computational properties, in:
  J.-F. Sowa (Ed.), Principles of {S}emantics {N}etworks, Morgan Kaufmann,
  1991, pp. 331--361.

\bibitem{Tarski55}
A.~Tarski, A lattice-theoretical fixpoint theorem and its applications, Pacific
  Journal of Mathematics 5 (1955) 285--309.

\bibitem{Kurz00}
A.~Kurz, A co-variety-theorem for modal logic, in: M.~Zakharyaschev,
  K.~Segerberg, M.~de~Rijke, H.~Wansing (Eds.), Advances in Modal Logic,
  {V}olume 2, {CSLI} Publications, 2000, pp. 367--380.

\end{thebibliography}

\end{document}